\documentclass[conference]{IEEEtran}
\IEEEoverridecommandlockouts
\usepackage{bm}
\usepackage{amsthm} 
\newtheorem{theorem}{Theorem}

\newtheorem{remark}{Remark}
\newtheorem{lemma}{Lemma}
\newtheorem{definition}{Definition}
\newtheorem{assumption}{Assumption}
\newtheorem{proposition}{Proposition}
\usepackage{algorithm}
\usepackage{algpseudocode} 
\usepackage{amsmath,amsfonts}
\allowdisplaybreaks[4]
\usepackage{array}
\usepackage[caption=false,font=normalsize,labelfont=sf,textfont=sf]{subfig}
\usepackage{textcomp}
\usepackage{stfloats}
\usepackage{url}
\usepackage{verbatim}
\usepackage{graphicx}
\usepackage{cite}
\usepackage{amssymb}
\usepackage{nomencl}
\makenomenclature
\usepackage{etoolbox}
\usepackage{url}
\usepackage{booktabs} 
\usepackage{xcolor}

\newcommand{\xinyiafterreview}[1]{#1}
\newcommand{\xinyiafterrevieww}[1]{#1}

\def\BibTeX{{\rm B\kern-.05em{\sc i\kern-.025em b}\kern-.08em
    T\kern-.1667em\lower.7ex\hbox{E}\kern-.125emX}}

\makeatletter
\newcommand{\linebreakand}{%
  \end{@IEEEauthorhalign}
  \hfill\mbox{}\par
  \mbox{}\hfill\begin{@IEEEauthorhalign}
}
\makeatother

\begin{document}

\title{On-Policy Reinforcement-Learning Control for Optimal Energy Sharing and Temperature Regulation in District Heating Systems\\
}

\author{Xinyi Yi and Ioannis Lestas%
    \thanks{X. Yi and I. Lestas are with the Department of Engineering, University of Cambridge, Trumpington Street, Cambridge, CB2 1PZ, United Kingdom. Emails:
        {\tt\small <xy343, icl20>@cam.ac.uk}.}
}
\maketitle

\begin{abstract}
We address the problem of temperature regulation and optimal energy sharing in district heating systems (DHSs) where the demand and system parameters are unknown. We propose a temperature regulation scheme that employs data-driven on-policy updates that achieve these objectives. In particular, we show that the proposed control scheme converges to an optimal equilibrium point of the system, while also having guaranteed convergence to an optimal LQR control policy, thus providing good transient performance. The efficiency of our approach is also demonstrated through extensive simulations.
\end{abstract}

\section{Introduction}
\xinyiafterreview{Carbon neutrality places the heating sector at the center of emission reduction. District Heating Systems (DHSs), common in China, Russia, and Europe, can decarbonize heating by integrating low-carbon sources such as waste heat and large-scale heat pumps \cite{liu2024diversifying}. This transition cuts fossil fuel use but creates complex generation profiles, requiring advanced DHS controllers for efficiency.}

\xinyiafterreview{One challenge in DHSs with multiple energy sources is that the optimal equilibrium points for energy sharing are not known a priori, as they depend on the generally unknown demand profile. Traditional regulators and model predictive controllers require disturbance forecasting, risking suboptimality, while optimization-feedback methods often assume unrealistically fast dynamics not present in DHSs.} Furthermore, unknown system parameters, for example due to changing occupancy levels, complicates controller design. Our aim is to design controllers that achieve both of these objectives. This is achieved by combining transformations that achieve optimal equilibirum points and transient performance without a prior knowledge of demand, which have been used thus far in model-based schemes \cite{accconfer}, with data-based approaches that achieve asymptotically optimal transient performance.

Research in the area of data-driven Linear Quadratic Regulator (LQR) has evolved into two main approaches: indirect and direct controllers. Unlike indirect controllers, which require system identification, direct controllers circumvent this requirement. \xinyiafterreview{Recent advances in direct methods include reinforcement learning (RL), policy gradient optimization (PGO), and certainty-equivalence (CE). With greedy policy updates, RL often converges faster in deterministic settings. This work employs online RL for its model-free, real-time adaptability, while extensions to stochastic settings with disturbances and uncertainties are left for future work.} The papers \cite{6315769,735224} introduced adaptive online RL controllers, while \cite{lopez2023efficient} and \cite{hao2024quadratic} advanced offline RL strategies, with \cite{hao2024quadratic} specifically addressing Gaussian disturbances. These works employ Q-function estimation to ensure policy convergence.

Recent data-driven temperature controller mainly focused on optimizing building regulation \cite{cholewa2022easy}. Temperature regulation differs substantially between individual buildings and DHSs in two key aspects: 1) Modeling complexity—building models vary widely to capture unique structures, whereas DHSs rely on standardized components such as pipes and heat exchangers; 2) Energy sharing—building control balances cost and comfort for individuals, while DHSs emphasize efficient energy distribution across multiple consumers and generators.

The contribution of this paper is the introduction of an on-policy RL temperature regulation scheme for DHSs that achieves both optimal energy sharing with minimal steady-state temperature deviations and asymptotically optimal transient performance. The proposed method adapts online with guaranteed convergence to an optimal LQR controller and equilibrium point, independent of prior knowledge of disturbances or system models. Moreover, we prove that during online adaptation, the controller ensures stability and convergence to the optimal equilibrium point. 

The paper is structured as follows. The steady-state optimal energy sharing and temperature regulation problem are formulated in Section~II. A data-driven RL regulator with asymptotically optimal performance is presented in Section~III. Simulations validating the scheme are given in Section~IV, and key findings are summarized in Section~V.

\section{Preliminaries}
This section provides an overview of
the temperature dynamics and the economic dispatch problem that needs to be solved at steady-state within a DHSs.

\subsection{Temperature dynamics and economic dispatch in DHSs}
\begin{assumption}\textbf{(Properties of DHSs)\cite{accconfer}}

\begin{enumerate}
    \item Heat conduction is not considered in the model.
    \item Flow rates are assumed to be constant and are externally controlled by users. The paper does not address the control of mass flow rates.
    \item The water's density \(\rho\) and specific heat capacity \(c_{\text{s.h.}}\) are uniform and constant throughout the system.
    \item The network is designed to be leak-free and experiences no hydraulic losses.
    \item The supply and return networks are symmetric.
\end{enumerate}
\end{assumption}
\subsubsection{Temperature dynamics}
A DHS consists of supply and return networks with HXs, where consumers extract heat and return cooled water, and producers reheat and recirculate it as shown in Fig.~\ref{proandcon}. For HX \(i\), the supply and return temperatures are \(T_{i}^{S}\) and \(T_{i}^{R}\), the flow rates are \(q_{i}\) and \(q_i^{E}\), and all temperatures denote deviations from steady-state values under a nominal demand profile.

\vspace{-0.9cm}
  \begin{figure}[htbp]
	\vspace{0.5cm}
	\hspace{-0.1cm}
	\includegraphics[width=3.1in]{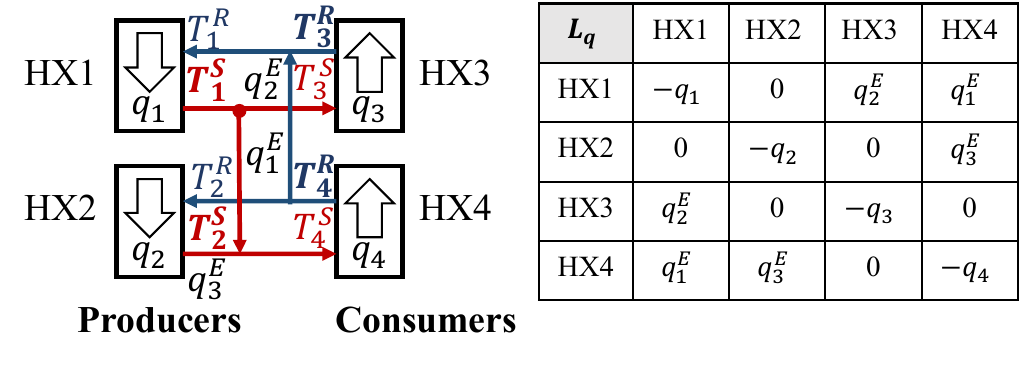}
	\vspace{-0.5cm}
	\caption{{\small {An example of a DHS, where HX 1 supplies heat for both HX 3 and HX 4, while HX2 only supplies heat for HX 4. The table gives the entries of the corresponding adjacency matrix $\boldsymbol{L_q}$ in \eqref{e1a}.}}}
	\label{proandcon}
\end{figure}
\vspace{-0.3cm}

The temperature dynamics of HX \(i\), as modelled in \cite{machado2022decentralized} under \textbf{Assumption 1}, are given as follows,
\begin{equation}\label{pipedy}
\rho c_{\text{s.h.}} V_i \dot{T}_{i}^{\text{out}} = \rho c_{\text{s.h.}} q_{i} (T_{i}^{\text{in}} - T_{i}^{\text{out}}) + P_i + P_i^{\text{dis}},
\end{equation}
where \( V_i \) is the volume of HX $i$, \(P_i\) is the controllable thermal power of HX \(i\), \(T_{i}^{\text{in}}\) and \(T_{i}^{\text{out}}\) are the inlet and outlet temperatures of HX \(i\), respectively. If HX \(i\) functions as a producer, \(P_i\) denotes the heat production power, \(T_{i}^{\text{in}}=T_i^R\) and \(T_{i}^{\text{out}}=T_i^S\). If HX \(i\) serves as a consumer, \(P_i\) denotes the adjustable heat load power, \(T_{i}^{\text{in}}=T_i^S\) and \(T_{i}^{\text{out}}=T_i^R\).

HX $i$'s inlet  $T_{i}^{in}$ is determined by the energy balance equation given as follows:
\begin{equation}\label{pipeenergybalance}
 0 = \sum_{j \in \mathcal{O}^{HX}_k \cap \mathcal{I}^{HX}_i}  q_j^E T_{k}^{\text{out}} - T_{i}^{in} \sum_{j \in \mathcal{O}^{HX}_k \cap \mathcal{I}^{HX}_i}  q_j^E,
\end{equation}
where $\mathcal{O}^{HX}_i$ and $\mathcal{I}^{HX}_i$ denote the sets of pipelines associated with the outlet and inlet of HX $i$, respectively.

By calculating the inlet temperature $T_{i}^{in}$ from the outlet temperature $T_{i}^{out}$ via equation (\ref{pipeenergybalance}), the temperature dynamics (\ref{pipedy}) can be described by the outlet temperature vector $\boldsymbol{T}$ as:
\begin{equation}\label{e1a}
\boldsymbol{V\dot{T}} = \boldsymbol{L_qT + P+ P^{dis} },
\end{equation}
where $\boldsymbol{P}$ is the vector of controllable load and production of HXs(i.e. it is the controlled variable).
\xinyiafterreview{$\boldsymbol{P^{dis}}$ denotes constant disturbances representing slow demand changes \cite{machado2022decentralized}.} Both $\boldsymbol{P}$ and $\boldsymbol{P^{dis}}$ are weighted by
$\frac{1}{\rho c_{\text{s.h.}}}$. $\boldsymbol{V}$ is the diagonal matrix of HXs' volumes. $\boldsymbol{L_q}$ is the adjacency matrix of flow rates between HXs, defined by:
\begin{equation}\label{Lqdefi}
\boldsymbol{L_q}(i,j) =
\begin{cases}
q^E_k, & (i \in \mathcal{G} \text{ and } j \in \mathcal{L}) or (i \in \mathcal{L} \text{ and } j \in \mathcal{G}), \\
-q_{i} & i=j \text{ and } i \in \mathcal{G} \cup \mathcal{L},\\
0, & \text{otherwise},
\end{cases}
\end{equation}
\noindent where $\mathcal{L}$ and $\mathcal{G}$ are the sets of consumer and producer HXs, respectively. $q^E_k$ is the flow rate through pipe $k$, connecting HX $i$ and HX $j$.

\subsubsection{Economic dispatch}
The economic dispatch for DHSs (\ref{e1a}) at steady-state, is modeled as two optimization problems.
\begin{subequations}\label{opt1}
\begin{align}
\textbf{E1:} &\ \min_{\boldsymbol{P}\in \mathbb{R}, \boldsymbol{T}\in \mathbb{R}} \Phi_1 = \sum \frac{1}{2} f_{i} (P_i)^2 = \frac{1}{2} \boldsymbol{P}^{\top} \boldsymbol{F} \boldsymbol{P},\\
    &\ \text{s.t.} \ \boldsymbol{L_q T} = -(\boldsymbol{P^{dis}} + \boldsymbol{P}),\label{5b}
\end{align}
\end{subequations}
\xinyiafterreview{where $\Phi_1$ denotes the cost of production and load adjustments, with \( \boldsymbol{F}=\mathrm{diag}\{ f_{i}\} \) and \( f_{i}>0 \) the cost coefficient at node \( i \). Since $\boldsymbol{F}$ is positive definite, $\Phi_1$ is quadratic and strictly convex. The linear constraints in \eqref{5b} make \textbf{E1} a convex optimization problem with a unique solution $\boldsymbol{P^*}$. As $\boldsymbol{L_q}$ is a Kirchhoff matrix \cite{hangos1999thermodynamic}, an optimal $\boldsymbol{T^*}$ satisfies  
$\boldsymbol{T^*} = -\boldsymbol{L_q^\dagger}(\boldsymbol{\bar{P}^{dis}} + \boldsymbol{P^*}) + z\boldsymbol{1}, \quad z \in \mathbb{R}$, 
where \( z\boldsymbol{1} \) lies in the nullspace of \( \boldsymbol{L_q} \) and $\boldsymbol{L_q^\dagger}$ is its Moore–Penrose pseudoinverse. The set of all such $\boldsymbol{T^*}$ is convex.}
 \footnote{\xinyiafterreview{Note that $\boldsymbol{T^*}$ is not unique; therefore, E2 minimizes the temperature deviation (see Lemma 1 in \cite{accconfer}). }}

The temperature deviation cost
is minimized in problem \textbf{E2}.
\begin{equation}\label{economic2}
\textbf{E2:} \ \min_{z \in \mathbb{R}, \boldsymbol{T \in \mathbb{R}}} \frac{1}{2} \boldsymbol{T}^\top \boldsymbol{G} \boldsymbol{T}, \quad s.t. \boldsymbol{T} = -\boldsymbol{L_q^\dagger} (\boldsymbol{\bar{P}^{dis}} + \boldsymbol{{P}^*}) + z\boldsymbol{1}.
\end{equation}
where \( \boldsymbol{G}=diag\{ g_{i}\} \), and \( g_{i}>0 \) represents the cost coefficient at node \( i \).

\begin{theorem}
If the DHS (\ref{e1a}) achieves equilibrium at $\boldsymbol{{T^*}^*}$ and $\boldsymbol{P}^*$, and satisfies $\boldsymbol{F^M P}^* = \boldsymbol{0}$ and $\boldsymbol{1^\top G {T^*}^*} = 0$, where \xinyiafterreview{$\boldsymbol{{T^*}^*}$ denotes the optimal temperature set for E1 and E2, and}  $\boldsymbol{F^M}$
is defined by the following matrix: {\small \begin{equation}\label{fmdifinition}
\begin{bmatrix}
F(1,1) & -F(2,2) & 0 & \cdots & 0\\
0 & F(2,2) & -F(3,3) & \cdots & 0\\
\vdots & \vdots & \vdots & \ddots & \vdots\\
\end{bmatrix},
\end{equation}} then it solves the optimization problems \textbf{E1} and \textbf{E2}. \footnote{Note that the optimality conditions are model-independent, providing a basis for designing a data-driven energy-sharing temperature controller.}.
\end{theorem}
The proof can be found in Theorems 1-4 of \cite{accconfer}.

\subsection{Problem formulation for optimal steady state}
\subsubsection{Discrete-time Temperature Dynamics}
We discretize the continuous-time DHS dynamics \eqref{e1a} using Zero-Order Hold with Euler’s approximation. The data-based controller also applies to other discretization methods. With sampling period \( \tau \) and \( \boldsymbol{E} = \frac{\boldsymbol{I}}{\boldsymbol{V}} \), the dynamics become:
\begin{equation}\label{disheat}
\boldsymbol{T_{k+1}} = \boldsymbol{(I +} \tau \boldsymbol{E L_q) T_{k}} + \tau \boldsymbol{E P_k} + \tau \boldsymbol{E P_k^{dis}}.
\end{equation}

We assume that \(-1\) is not an eigenvalue of \(\tau\boldsymbol{EL_q}\), which implies the controllability of \(((\boldsymbol{I} + \tau \boldsymbol{EL_q}) \boldsymbol{T_k}, \tau \boldsymbol{E})\) in (\ref{disheat}).

\subsubsection{Output and error definition}
We aim to design a temperature regulator ensuring convergence to an optimal equilibrium point $(\boldsymbol{{T^*}^*,P^*})$ as defined in \textbf{E1, E2}, respectively. To achieve this we consider an error signal which when driven to zero ensures that the corresponding optimality conditions are satisfied. Therefore, the output and error definitions are derived from the \textbf{E1-E2} optimality conditions in \textbf{Theorem 1}:
\begin{equation}\label{13}
y_k = \boldsymbol{1^\top G T_k},  \boldsymbol{e_k} = \boldsymbol{\Lambda_p P_k + \Lambda_y y_k},
\end{equation}
where $\boldsymbol{\Lambda_p} = \begin{bmatrix} \boldsymbol{F^M} \\ 0 \end{bmatrix}, \boldsymbol{\Lambda_y} = \begin{bmatrix} \boldsymbol{0} \\ 1 \end{bmatrix}$.

\subsubsection{Augmented Dynamics}
From \textbf{Theorem 1}, if $\boldsymbol{e_k = 0}$, the DHS reaches its optimal state. At steady state, $\boldsymbol{\delta T_k = T_k - T_{k-1} = 0}$ and $\boldsymbol{\delta u_k = P_k - P_{k-1} = 0}$. Furthermore, subtracting adjacent time terms from \eqref{disheat} and \eqref{13} eliminates disturbances, yielding the augmented dynamics:
\begin{subequations}\label{timingterms}
\begin{align}
\boldsymbol{\epsilon_{k+1}} &= \boldsymbol{A_\epsilon \epsilon_k + B_\epsilon \delta u_k}, \label{timingtermsa}\\
\boldsymbol{e_k} &= \boldsymbol{C_\epsilon \epsilon_k + D_\epsilon \delta u_k},\label{timingtermsb}
\end{align}
\end{subequations}
where the matrices are defined as: $\boldsymbol{A_\epsilon}= \begin{bmatrix}
\boldsymbol{(I +} \tau \boldsymbol{E L_q)} & \boldsymbol{0} \\
\boldsymbol{\Lambda_y 1^\top G} & \boldsymbol{I}
\end{bmatrix},$
$\boldsymbol{B_\epsilon} = \begin{bmatrix}
\tau \boldsymbol{E} \\
\boldsymbol{\Lambda_p}
\end{bmatrix},
\boldsymbol{C_\epsilon} = \boldsymbol{[\Lambda_y 1^\top G, I]},\boldsymbol{D_\epsilon} = \boldsymbol{\Lambda_p}$,
 $\boldsymbol{\epsilon_k} = \begin{bmatrix}
\boldsymbol{\delta T_k}\\
\boldsymbol{e_{k-1}}
\end{bmatrix}$
$\boldsymbol{\delta T_k} = \boldsymbol{T_k}\boldsymbol{- T_{k-1}, \delta u_k} = \boldsymbol{P_k - P_{k-1}}$.
\footnote{\xinyiafterreview{In the augmented system, the state variable $\boldsymbol{\epsilon_k}$ has dimension $n=2n_{T_k}$ and the input variable $\boldsymbol{\delta T_k}$ has dimension $m=n_{T_k}$. For a fixed input, the steady-state equilibrium is uniquely determined by the network balance; hence, driving the error state to zero ensures convergence to the corresponding optimal equilibrium. Notably, the error depends only on observable ratios of state and input variables, without requiring explicit knowledge of $\boldsymbol{L_q}$ or $P^{dis}$.}}

\begin{assumption}
The matrix \(\boldsymbol{\begin{bmatrix}
    \tau \boldsymbol{E L_q} & \tau \boldsymbol{E} \\
    \boldsymbol{\Lambda_y 1^\top G} & \boldsymbol{\Lambda_p}
    \end{bmatrix}}\) has full row rank.
This assumption ensures \(\boldsymbol{(A_\epsilon, B_\epsilon)}\) is controllable.\cite{asuk2021feedback}\footnote{\xinyiafterrevieww{If the DHS (\ref{disheat}) is stabilizable, then the first \(n_{\boldsymbol{T}_k}\) rows are linearly independent and the second \(n_{\boldsymbol{T}_k}\) rows are linearly independent due to distinct column entries; consequently, the matrix is generally of full row rank \(2n_{\boldsymbol{T}_k}\).}}
\end{assumption}

\subsubsection{Control problem definition}
\xinyiafterreview{The challenge of an unknown system matrix stems from variations in the \(\boldsymbol{L_q}\) coefficients with changing mass flow rates. For temperature regulation, we assume constant mass flow, since hydraulic dynamics evolve much faster than temperature dynamics.}

\begin{proposition}
When the augmented dynamic system (\ref{timingterms}) reaches the equilibrium point where \(\boldsymbol{\delta u_k^* = 0}\), \(\boldsymbol{\delta T_k^* = 0}\) and \(\boldsymbol{e_k^* = 0}\), the augmented state \(\boldsymbol{\epsilon_k^*=0}\). Consequently, the original DHS described by equations (\ref{disheat}, \ref{13}) is at its optimal equilibrium point \(\boldsymbol{({T^*}^*, P^*)}\), as defined in (\ref{opt1}, \ref{economic2}).
\end{proposition}

\begin{proof}
The fact that the
the identity matrix \( \boldsymbol{I} \) is included in \( \boldsymbol{C_\epsilon} \) guarantees that the matrix \( \boldsymbol{C_\epsilon} \) has full row rank. At the equilibrium point where \( \boldsymbol{e_k=0} \), \( \boldsymbol{\delta u_k = P_k - P_{k-1} = 0} \), this leads to \( \boldsymbol{\epsilon_k = 0} \), and consequently, \( \boldsymbol{\delta T_k = 0} \). Under these conditions, optimality is confirmed based on \textbf{Theorem 1}, as all the corresponding optimality conditions are satisfied.
\end{proof}

\xinyiafterreview{We propose a data-driven LQR controller of the form \(\boldsymbol{\delta u_k} = -\boldsymbol{K_k \epsilon_k}\), which ensures convergence to the optimal equilibrium \((\boldsymbol{\delta u_k^* = 0, e_k^* = 0, \epsilon_k^* = 0})\) without prior knowledge of disturbances or system matrices. During transients, the controller is updated with guaranteed convergence to an optimal LQR controller for system \eqref{timingterms}.}

\section{On-policy RL temperature regulator}
\xinyiafterreview{Given uncertainties in model parameters such as mass flow rates, we develop an on-policy RL controller integrated with the Section II framework to ensure steady-state optimality.}

\subsection{Value function and Q function}
\subsubsection{Value function}
\xinyiafterreview{We optimize the trajectory \((\boldsymbol{e_k}, \boldsymbol{\delta u_k})\) of \eqref{timingterms} using a quadratic cost with positive definite weights \(\boldsymbol{Q_e}\) and \(\boldsymbol{R_e}\). For input \(\boldsymbol{\delta u_k} = -\boldsymbol{K \epsilon_k}\), the value function \(V^{\boldsymbol{K}}\) denotes the cost from time \(k\) onward under \(\boldsymbol{\delta u_i} = -\boldsymbol{K \epsilon_i}\).}
\begin{subequations}\label{overallcost}
    \begin{align}
     V^{\boldsymbol{K}}(\boldsymbol{e_k,\delta u_k}):
     = &\sum_{i=k}^\infty c(\boldsymbol{e_i},\boldsymbol{\delta u_i}) \text{ s.t. } \boldsymbol{\delta u_i}=\boldsymbol{-K\epsilon_i}\\
     =& \frac{1}{2}\sum_{i=k}^\infty \left(\boldsymbol{e_{i}^\top Q_e e_{i} + \delta u_{i}^\top R_e \delta u_{i}}\right).\label{19b}
    \end{align}
\end{subequations}

Substituting $\boldsymbol{e_i}$ with  $\boldsymbol{\epsilon_i}$ based on  (\ref{timingtermsb}), we can obtain:
{\small \begin{subequations}\label{overallcost1}
    \begin{align}
     V^{\boldsymbol{K}}(\boldsymbol{\epsilon_k}):
      =&\sum_{i=k}^\infty c(\boldsymbol{\epsilon_i},\boldsymbol{\delta u_i}) \text{ s.t. } \boldsymbol{\delta u_i}=\boldsymbol{-K\epsilon_i}\\
      =& \frac{1}{2}\sum_{i=k}^\infty(\boldsymbol{\epsilon_{i}^\top Q_\epsilon \epsilon_{i} + \epsilon_{i}^\top N_\epsilon \delta u_{i}+ \delta u_{i}^\top N_\epsilon^\top \epsilon_{i}}\\\notag
      &+\boldsymbol{ \delta u_{i}^\top }\boldsymbol{ R_\epsilon  \delta u_{i}})\\
    =&\frac{1}{2} \sum_{i=k}^\infty \boldsymbol{\begin{bmatrix}
        \epsilon_i,\delta u_i
    \end{bmatrix}\bar{Q}\begin{bmatrix}
        \epsilon_i\\
        \delta u_i
    \end{bmatrix}}\label{20cc}\\
    =& \frac{1}{2}\boldsymbol{\epsilon_k}^\top\boldsymbol{Q_\text{eff}\epsilon_k}+V^{\boldsymbol{K}}(\boldsymbol{\epsilon_{k+1}}) \ =\frac{1}{2}\boldsymbol{\epsilon_k}^\top \boldsymbol{P^K}\boldsymbol{\epsilon_k},\label{20cd}
    \end{align}
\end{subequations}}

\noindent where $\boldsymbol{Q_\epsilon = C_\epsilon^\top Q_e C_\epsilon}$, $\boldsymbol{ N_\epsilon = C_\epsilon^\top Q_e \Lambda_u}, \boldsymbol{R_\epsilon = R_e + \Lambda_u^\top }$ $\boldsymbol{Q_e \Lambda_u}$, $\boldsymbol{Q_{\text{eff}}}= \boldsymbol{Q_\epsilon-N_\epsilon K-K^\top N_\epsilon^\top}\boldsymbol{+K^\top R_\epsilon K}$ and $\boldsymbol{\bar{Q}}=\begin{bmatrix}
        \boldsymbol{Q_\epsilon}&\boldsymbol{N_\epsilon}\\
       \boldsymbol{N_\epsilon^\top}&\boldsymbol{R_\epsilon}
    \end{bmatrix}$. We examine the properties of the matrices associated with the value function, specifically \(\boldsymbol{P^K}\), \(\boldsymbol{\bar{Q}}\), and \(\boldsymbol{Q_\text{eff}}\), to facilitate further derivations. With positive definite matrices \(\boldsymbol{Q_e}\) and \(\boldsymbol{R_e}\), the value function $V^{\boldsymbol{K}}(\boldsymbol{e_k}, \boldsymbol{\delta u_k}) = V^{\boldsymbol{K}}(\boldsymbol{\epsilon_k}) = \frac{1}{2} \boldsymbol{\epsilon_k}^\top \boldsymbol{P^K} \boldsymbol{\epsilon_k} \geq 0$. This function only equals zero when \(\boldsymbol{e_i = 0}\), \(\boldsymbol{\delta u_i = 0}\), and \(\boldsymbol{\epsilon_i = 0}\) from step \(k\) onward. In other words, \(\frac{1}{2} \boldsymbol{\epsilon_k}^\top \boldsymbol{P^K} \boldsymbol{\epsilon_k} = 0\) occurs solely when \(\boldsymbol{\epsilon_i = 0}\). Thus, \(\boldsymbol{P^K}\) is positive definite. Similarly, the positiveness of \(\boldsymbol{\bar{Q}}\) and \(\boldsymbol{Q_\text{eff}}\) are verified through the stage cost as given in
    equation \eqref{19b} and from \eqref{20cc} to \eqref{20cd}:
\begin{equation}
\begin{aligned}
\boldsymbol{\begin{bmatrix}
        \epsilon_k & \delta u_k
    \end{bmatrix} \bar{Q} \begin{bmatrix}
        \epsilon_k \\
        \delta u_k
    \end{bmatrix}}&= \boldsymbol{e_{k}^\top Q_e e_{k} + \delta u_{k}^\top R_e \delta u_{k}}\\
    &= \boldsymbol{\epsilon_k^\top Q_\text{eff} \epsilon_k \geq 0}.
\end{aligned}
\end{equation}
From \eqref{20cd}, $\boldsymbol{P^K}$ satisfies 
the discrete-time Lyapunov equation
\begin{equation}\label{bellman1}
\begin{aligned}
\boldsymbol{P^K} = \boldsymbol{(A_\epsilon - B_\epsilon K)^\top P^K (A_\epsilon - B_\epsilon K) }\boldsymbol{+Q_\text{eff}}.
\end{aligned}
\end{equation}
This is known to have a unique positive definite solution $\boldsymbol{P^K}$ when $\boldsymbol{(A_\epsilon - B_\epsilon K)}$ is stable.

\subsubsection{Q function}
The Q-function under policy $\boldsymbol{K}$, denoted $Q^{\boldsymbol{K}}$, is the cost of applying an arbitrary input $\boldsymbol{\delta u_k}$ at time $k$, followed by $\boldsymbol{\delta u_i} = -\boldsymbol{K \epsilon_i}$ for $i > k$.
\begin{equation}\label{qfunction}
  Q^K\boldsymbol{(\boldsymbol{\epsilon_k}, \boldsymbol{\delta u_k})} = c(\boldsymbol{\epsilon_k}, \boldsymbol{\delta u_k}) + \boldsymbol{V^K(\epsilon_{k+1}}).
\end{equation}
It is important to note that the functions \(\boldsymbol{Q^K}\) (\ref{qfunction}) and \(\boldsymbol{V^K}\) (\ref{overallcost1}) are related as \(\boldsymbol{Q^K(\boldsymbol{\epsilon_k}, -\boldsymbol{K \epsilon_k})} = \boldsymbol{V^K(\boldsymbol{\epsilon_k})}\).
Substituting equations (\ref{timingterms}) and (\ref{overallcost1}) into (\ref{qfunction}), we get
\begin{subequations}\label{recur}
\begin{align}
 \boldsymbol{Q^K(\boldsymbol{\epsilon_k},\boldsymbol{\delta u_k})}
    =&\frac{1}{2}{\begin{bmatrix} \boldsymbol{\epsilon_k}\\\boldsymbol{\delta u_{k}}\end{bmatrix}}^\top
     \boldsymbol{\Theta^K} \begin{bmatrix}
     \boldsymbol{\epsilon_{k}}\\    \boldsymbol{\delta u_{k}}
    \end{bmatrix}\\
    =&  \frac{1}{2} {\begin{bmatrix} \boldsymbol{\epsilon_k}\\\boldsymbol{\delta u_{k}}\end{bmatrix}}^\top\begin{bmatrix}
        \boldsymbol{Q_\epsilon}&\boldsymbol{N_\epsilon}\\
        \boldsymbol{N_\epsilon^T}&\boldsymbol{R_\epsilon}
    \end{bmatrix}\begin{bmatrix} \boldsymbol{\epsilon_k}\\\boldsymbol{\delta u_{k}}\end{bmatrix}+\frac{1}{2}\notag\\
    &{\begin{bmatrix} \boldsymbol{\epsilon_{k+1}}\\\boldsymbol{-K\epsilon_{k+1}}\end{bmatrix}}^\top\boldsymbol{\Theta^K} \begin{bmatrix}
     \boldsymbol{\epsilon_{k+1}}\\    \boldsymbol{-K\epsilon_{k+1}}
    \end{bmatrix},\label{30c}
\end{align}
\end{subequations}
where $\boldsymbol{\Theta^K}=\begin{bmatrix} \boldsymbol{Q_\epsilon+ A_\epsilon^T P^K A_\epsilon}&\boldsymbol{N_\epsilon+ A_\epsilon^T P^K B_\epsilon}\\
\boldsymbol{N_\epsilon^T+ B_\epsilon^T P^K A_\epsilon}&\boldsymbol{R_\epsilon+ B_\epsilon^T P^K B_\epsilon}
\end{bmatrix}$.

$\boldsymbol{\Theta^K}$ can be partitioned as $\begin{bmatrix}
\boldsymbol{\Theta_{\epsilon\epsilon}}&    \boldsymbol{\Theta_{\delta u\epsilon}}^T\\
\boldsymbol{\Theta_{\delta u\epsilon}}&\boldsymbol{\Theta_{\delta u\delta u}}
\end{bmatrix}$, where $\boldsymbol{\Theta_{\epsilon\epsilon}}=\boldsymbol{Q_\epsilon+ A_\epsilon^T P^K A_\epsilon}$, $\boldsymbol{\Theta_{\delta u\epsilon}}=\boldsymbol{N_\epsilon^T+ B_\epsilon^T P^K A_\epsilon}$, $\boldsymbol{\Theta_{\delta u\delta u}}=\boldsymbol{R_\epsilon+ B_\epsilon^T P^K B_\epsilon}$.\footnote{Positive definite matrices $\boldsymbol{R_e}$, $\boldsymbol{Q_e}$, and $\boldsymbol{P^K}$ ensure the positive definiteness of $\boldsymbol{\Theta_{\delta u \delta u}}$, making it invertible.}

Substituting $\boldsymbol{\epsilon_{k+1}=A_\epsilon \epsilon_k+B_\epsilon \delta u_k}$, (\ref{recur}) is rewritten as:
\begin{equation}\label{rerecur}
    \boldsymbol{\Theta^K}=\boldsymbol{\bar{Q}}+ \boldsymbol{\phi_K}^\top \boldsymbol{\Theta^K}\boldsymbol{\phi_K},
\end{equation}
where $\boldsymbol{\phi_K}=\begin{bmatrix}
        \boldsymbol{A_\epsilon}&\boldsymbol{B_\epsilon}\\
        \boldsymbol{-KA_\epsilon}&\boldsymbol{-KB_\epsilon}
    \end{bmatrix}$.

\begin{lemma}
\textbf{(Existence and Uniqueness of \(\boldsymbol{\Theta^K}\) for a Stabilizing \(\boldsymbol{K}\)).}
Let \(\boldsymbol{K}\) be a stabilizing feedback gain, such that 
\(\rho(\boldsymbol{A_\epsilon} - \boldsymbol{B_\epsilon K}) < 1\). 
Then
\(\rho(\boldsymbol{\phi_K}) < 1\), ensuring the existence of a unique positive definite matrix \(\boldsymbol{\Theta^K}\), which satisfies \eqref{rerecur}.
\end{lemma}

\begin{proof}
Consider the transformation matrix:$\boldsymbol{T} = \begin{bmatrix}
\boldsymbol{I} & \boldsymbol{0} \\
-\boldsymbol{K} & \boldsymbol{I}
\end{bmatrix}$,$\boldsymbol{T^{-1}} = \begin{bmatrix}
\boldsymbol{I} & \boldsymbol{0} \\
\boldsymbol{K} & \boldsymbol{I}
\end{bmatrix}$.
Using this transformation, we compute $\boldsymbol{T^{-1}} \boldsymbol{\phi_K} \boldsymbol{T} =
\begin{bmatrix}
\boldsymbol{A_\epsilon - B_\epsilon K} & \boldsymbol{B_\epsilon} \\
\boldsymbol{0} & \boldsymbol{0}
\end{bmatrix}$. Since similarity transformations preserve eigenvalues, the eigenvalues of \(\boldsymbol{\phi_K}\) are the same as those of \(\boldsymbol{A_\epsilon - B_\epsilon K}\), augmented by additional zero eigenvalues. Given that \(\rho(\boldsymbol{A_\epsilon - B_\epsilon K}) < 1\), it follows that the spectral radius \(\rho(\boldsymbol{\phi_K}) = \rho(\boldsymbol{A_\epsilon - B_\epsilon K}) < 1\), confirming that all eigenvalues of \(\boldsymbol{\phi_K}\) lie strictly within the unit circle. Hence for any \(\boldsymbol{\bar{Q}} \succ 0\) the Lyapunov Equation \eqref{rerecur} admits a unique positive definite solution \(\boldsymbol{\Theta^K} \succ 0\).
\end{proof}

\subsubsection{Minimization of Q function over the control policy $\boldsymbol{K}$} With the control policy $\boldsymbol{\delta u_i} = -\boldsymbol{K} \boldsymbol{\epsilon_i}$ for steps $i\geq k+1$, the control input $\boldsymbol{\delta u_k} = -\boldsymbol{K_k} \boldsymbol{\epsilon_k}$ at time step $k$ is determined by solving the minimization problem over the Q-function:
\begin{equation}\label{minimizationpro}
    \min_{\boldsymbol{K_k}}  \boldsymbol{Q^K(\epsilon_k, \delta u_k)}=\frac{1}{2}\begin{bmatrix}
     \boldsymbol{\epsilon_{k}}^\top & \boldsymbol{\delta u_{k}}^\top
    \end{bmatrix} \boldsymbol{\Theta^K} \begin{bmatrix}
     \boldsymbol{\epsilon_{k}} \\ \boldsymbol{\delta u_{k}}
    \end{bmatrix}.
\end{equation}
Solving for $\boldsymbol{K_k}$ results in $\boldsymbol{\delta u_k} = -\boldsymbol{\Theta_{\delta u\delta u}^{-1} \Theta_{\delta u\epsilon}} \boldsymbol{\epsilon_k} = -\boldsymbol{K_k} \boldsymbol{\epsilon_k}$ at time step $k$, and thus $\boldsymbol{K_k} = \boldsymbol{\Theta_{\delta u\delta u}^{-1} \Theta_{\delta u\epsilon}}$, where $\boldsymbol{\Theta_{\delta u\epsilon}}$ and $\boldsymbol{\Theta_{\delta u\delta u}}$ are defined in (\ref{recur}).

\subsection{On-policy RL temperature regulator for DHS}
\xinyiafterreview{We use Equation \eqref{recur} to compute $\boldsymbol{\Theta^{K_i}}$ from data collected along system trajectories, enabling an on-policy RL control strategy. In particular, we propose a data-driven policy iteration \textbf{Algorithm 1}, which guarantees convergence to the optimal control policy and the desired equilibrium point.} \footnote{An initial stabilizing controller can be derived from the nominal DHS model, which is generally straightforward due to the natural passivity of DHSs.}

\begin{algorithm}
\caption{RL adaptive temperature regulator}
\begin{algorithmic}[1]
\State \textbf{Initialize:} Select an initial stabilizing feedback policy $\boldsymbol{\delta u_k = -K_0 \epsilon_k}$.

\While{true}
    \State \textbf{Iteration $i$:}
    \State \textbf{Data collection:} At time $k$, apply the current policy $\boldsymbol{\delta u_k = -K_i \epsilon_k+es_k}$ for $N=(n_{\boldsymbol{\epsilon_k}}+1)n_{\boldsymbol{\delta u_k}}+n_{\boldsymbol{\epsilon_k}}$ intervals, with $\boldsymbol{es_{k}}$ as the probing noise to guarantee persistence of excitation of order $n_{\boldsymbol{\epsilon_k}} + 1$ (Definition \ref{def:PE}).
    \State \textbf{Equations construction:} Using $\boldsymbol{z_{k}}=\begin{bmatrix}
     \boldsymbol{\epsilon_{k}}\\
     \boldsymbol{\delta u_{k}}
    \end{bmatrix}$  constructed from the collected data, define $N$ equations:
    \begin{equation}\label{estimation}
    \boldsymbol{z_{k}^\top \hat{\Theta}^{K_{i+1}} z_{k} = z_{k}^\top \bar{Q} z_{k}} + \boldsymbol{\zeta_{i, k+1}^\top \hat{\Theta}^{K_{i+1}} \zeta_{i, k+1}},
    \end{equation}
where $\boldsymbol{\zeta_{i, k+1}} =
\begin{bmatrix}
\boldsymbol{\epsilon_{k+1}} \\
\boldsymbol{-K_i \epsilon_{k+1}}
\end{bmatrix}$. Then solve the set of $N$ equations for $\boldsymbol{\hat{\Theta}^{K_{i+1}}}$.
    \State \textbf{Policy improvement:} Improve the controller:
    \begin{equation}
     \boldsymbol{K_{i+1} = [\hat{\Theta}_{\delta u\delta u}^{K_{i+1}}] ^{-1}\hat{\Theta}_{\delta u\epsilon}^{K_{i+1}}}.
    \end{equation}
    \State Set $i = i + 1$ and go to iteration $i$.
\EndWhile $\boldsymbol{\|K_{i+1} - K_i\|} < \varepsilon$ for some $\varepsilon > 0$, recording $\boldsymbol{\hat{K}^*=K_{i+1}}$.

\State Adopt the optimal feedback policy $\boldsymbol{\delta u_k = -\hat{K}^* \epsilon_k}$.
\end{algorithmic}
\end{algorithm}

\begin{definition}\textbf{(Persistent order (PE)\cite{lopez2023efficient}).} \label{def:PE}
Consider a system $\boldsymbol{x_{k+1} = Ax_k + Bu_k}$, where $\boldsymbol{u_k}$ is the input vector with dimension $m$ and $\boldsymbol{x_k}$ is the state vector with dimension $n$. Let $N \geq (m+1)L - 1$. An input sequence $\{\boldsymbol{u_k}\}_{k=0}^{N-1}$ is said to be persistently exciting (PE) of order $L$, if the Hankel matrix $H_L(\boldsymbol{u}_{[0,N-1]})$ has a full rank of $mL$.
The Hankel matrix is constructed as follows:
\small{\[
H_L(\boldsymbol{u}_{[0,N-1]}) =
\begin{bmatrix}
\boldsymbol{u_0} & \boldsymbol{u_1} & \cdots & \boldsymbol{u_{N-L}} \\
\boldsymbol{u_1} & \boldsymbol{u_2} & \cdots & \boldsymbol{u_{N-L+1}} \\
\vdots & \vdots & \ddots & \vdots \\
\boldsymbol{u_{L-1}} & \boldsymbol{u_L} & \cdots & \boldsymbol{u_{N-1}}
\end{bmatrix}.
\]}
\end{definition}

We note the following points on the algorithm:

\textbf{1. Data collection:}
For iteration $i$, we collect $N = (n_{\boldsymbol{\epsilon_k}} + 1) n_{\boldsymbol{\delta u_k}} + n_{\boldsymbol{\epsilon_k}}$ data samples by applying the PE input $\boldsymbol{\delta u_k}$ of order $L = n_{\boldsymbol{\epsilon_k}} + 1$ to the system. In the DHS (\ref{disheat}, \ref{13}), we have $n_{\boldsymbol{\epsilon_k}} = 2n_{\boldsymbol{T_k}}$ and $n_{\boldsymbol{\delta u_k}} = n_{\boldsymbol{T_k}}$. \footnote{The probing noise in the algorithm is designed to vanish as \(k \to \infty\) \cite{lopez2023efficient}.}

\textbf{2. Equation Construction:}
In the algorithm design, we ensure sufficient data for the calculation process of $\boldsymbol{\hat{\Theta}^{K_{i+1}}}$, as outlined in Equation (\ref{estimation}). There are $N$ equations, with each equation corresponding to the data from each time step.

\begin{lemma}\textbf{(Rank condition\cite{lopez2023efficient}).}
    Let the dataset $\{\boldsymbol{x_k}\}_{k=0}^{N-1}$ be gathered from a controllable system, where an input $\{\boldsymbol{u_k}\}_{k=0}^{N-1}$ as a PE of order $n + L'$ is applied, then we have:
\[
\text{rank}\begin{pmatrix}
H_1(x_{[0,N-L']}) \\
H_{L'}(u_{[0,N-1]})
\end{pmatrix} = L'm + n.
\]
\end{lemma}

\begin{lemma}
\textbf{(Rank Condition of DHS Data)}.
The rank of the dataset \(\boldsymbol{z_{[0,N-1]}=\{z_k\}_{k=0}^{N-1}}\) collected from a DHS in an iteration of {\bf Algorithm 1} is \(3n_{\boldsymbol{T_k}}\).
\end{lemma}

\begin{proof}
Given the data collection process, we set \( L' = 1 \), and thus deduce from \textbf{Lemma 2} that for a DHS:
\begin{align*}
\text{rank}(\boldsymbol{z_{[0,N-1]}})=&\text{rank}(\boldsymbol{H_1(z_{[0,N-1]})})\\
=&\text{rank} \begin{pmatrix}
H_1(\epsilon_{[0,N-1]}) \\
H_1(\delta u_{[0,N-1]})
\end{pmatrix}\\
=&\xinyiafterreview{m+n}=3n_{\boldsymbol{T_k}}.
\end{align*}
\end{proof}

\begin{lemma}\textbf{(A sequence of stabilizing controllers\cite{1098829}).}
Starting with a stabilizing controller $\boldsymbol{K_0}$, the algorithm will generate a sequence of stabilizing controllers $\boldsymbol{K_i}$.
\end{lemma}
It follows from \textbf{Lemmas 1} and \textbf{4} that, given an initial stabilizing controller $\boldsymbol{K_0}$, the sequence $\boldsymbol{K_i}$ generated by \textbf{Algorithm 1} uniquely corresponds to matrices $\boldsymbol{P^{K_i}}$ and $\boldsymbol{\Theta^{K_i}}$.

\begin{assumption}
\textbf{(Accurate Data Measurement)} Both $\boldsymbol{\epsilon_k}$ and $\boldsymbol{\delta u_k}$ are assumed to be measured accurately.
\end{assumption}

\begin{theorem}
\textbf{(Calculation of \(\boldsymbol{\Theta^{K_i}}\)).}
Using the data collected via \textbf{Algorithm 1} and under \textbf{Assumptions 1-3}, \(\boldsymbol{\hat{\Theta}^{K_{i+1}}}\) calculated from \eqref{estimation} equals \(\boldsymbol{\Theta^{K_i}}\) of \(\boldsymbol{K_i}\) in \eqref{rerecur}.\footnote{\xinyiafterreview{Since DHS temperature regulation is an LQR problem with a PE input sequence $\boldsymbol{\delta u_k}$, Theorem~2 guarantees accurate Q-function estimation, enabling reliable greedy policy updates and avoiding conservative gradient-descent methods like Proximal Policy Optimization and Soft Actor-Critic.}}
\end{theorem}

\begin{proof}
From \textbf{Lemma 3} we have that we can find a set of $3n_{\boldsymbol{T_k}}$ vectors $\boldsymbol{z_k=[\epsilon_k^\top, \delta u_k^\top]^\top}$ that are linearly independent.Equation (\ref{estimation}) can be rewritten as:
\begin{equation}\label{estimatere}
\boldsymbol{\hat{\Theta}^{K_{i+1}}} = \boldsymbol{\bar{Q}}+\boldsymbol{\phi_{K_i}}^\top \boldsymbol{\hat{\Theta}^{K_{i+1}}} \boldsymbol{\phi_{K_i}}.
\end{equation}
This equation is structurally identical to the one outlined in \eqref{rerecur}, which is satisfied by the true Q-function matrix \(\boldsymbol{\Theta^{K_i}}\). \textbf{Lemma 3} asserts that sufficient data is available to accurately identify \(\boldsymbol{\Theta^{K_i}}\). Additionally, \textbf{Lemma 1} establishs that \(\rho(\boldsymbol{\phi_{K_i}}))<1\),which ensures that the discrete-time Lyapunov equation has a unique positive definite solution. Therefore, due to the uniqueness property of the Lyapunov equation, we can conclude that \(\boldsymbol{\hat{\Theta}^{K_{i+1}}} = \boldsymbol{\Theta^{K_i}}\).
\end{proof}

\subsection{Algorithm convergence}
\xinyiafterreview{In the reminder of this section, we show that as $\varepsilon\to 0$ {\bf Algorithm 1}, $K_i$ is guaranteed to converge to an optimal LQR policy that minimizes the cost in \eqref{overallcost}. Every adaptive controller $\boldsymbol{K_i}$ generated ensures convergence to the optimal equilibrium point of the system.}

\begin{proposition}\textbf{(The Optimal Temperature Regulator \(\boldsymbol{K_*}\)).}
The optimal controller \(\boldsymbol{\delta u_k = -K_* \epsilon_k}\) that minimizes the infinite horizon cost in \eqref{overallcost} and its corresponding \(\boldsymbol{P^{K_*}}\) and \(\boldsymbol{\Theta_*}\) defined in \eqref{overallcost1}, \eqref{recur}, respectively, satisfy the following:
\begin{enumerate}
    \item \(\boldsymbol{P^{K_*} = (A_\epsilon - B_\epsilon K_*)^\top P^{K_*} (A_\epsilon - B_\epsilon K_*) + Q_{\text{eff}}}\),
    \item \(\boldsymbol{K_* = (R_\epsilon + B_\epsilon^\top P^{K_*} B_\epsilon)^{-1} (B_\epsilon^\top P^{K_*} A_\epsilon + Q_\epsilon)}\),
    \item \[
\boldsymbol{\Theta_*} = \begin{bmatrix}
\boldsymbol{Q_\epsilon + A_\epsilon^\top P^{K_*} A_\epsilon} & \boldsymbol{N_\epsilon + A_\epsilon^\top P^{K_*} B_\epsilon} \\
\boldsymbol{N_\epsilon^\top + B_\epsilon^\top P^{K_*} A_\epsilon} & \boldsymbol{R_\epsilon + B_\epsilon^\top P^{K_*} B_\epsilon}
\end{bmatrix}.
\]
\end{enumerate}
\end{proposition}

\begin{lemma}\textbf{(Decreasing \(\boldsymbol{\hat{\Theta}^{K_i}}\).)}\label{lem:dec_theta}
Let \(\boldsymbol{\hat{\Theta}^{K_i}}\) and \(\boldsymbol{\hat{\Theta}^{K_{i+1}}}\) be the solutions of (\ref{estimation}) at two consecutive iterations of Algorithm 1. The following properties hold:
\begin{subequations}
\begin{align}
&\boldsymbol{\hat{\Theta}^{K_i}} - \boldsymbol{\hat{\Theta}^{K_{i+1}}} =
\sum_{j=0}^\infty
 \boldsymbol{\phi_{K_i}^\top}^j \boldsymbol{\Psi_{K_i}}^\top \boldsymbol{\hat{\Theta}^{K_i}} \boldsymbol{\Psi_{K_i}} \boldsymbol{\phi_{K_i}^j} \succeq 0,\label{aa}\\
&\boldsymbol{\hat{\Theta}^{K_{i+1}}} - \boldsymbol{\Theta_*} =
\sum_{j=0}^\infty
 \boldsymbol{\phi_{K_i}^\top}^j \boldsymbol{\Psi_*}^\top \boldsymbol{\Theta_*} \boldsymbol{\Psi_*} \boldsymbol{\phi_{K_i}^j} \succeq 0,\label{bb}
\end{align}
\end{subequations}
where
$\boldsymbol{\Psi_{K_i}} = \begin{bmatrix}
\boldsymbol{0} & \boldsymbol{0} \\
\boldsymbol{(K_i - K_{i-1})A_\epsilon} & \boldsymbol{(K_i - K_{i-1})B_\epsilon}
\end{bmatrix}$ and  \(\boldsymbol{\Psi_*} = \begin{bmatrix}
\boldsymbol{0} & \boldsymbol{0}\\
\boldsymbol{(K_* - K_{i+1}) A_\epsilon} & \boldsymbol{(K_* - K_{i+1}) B_\epsilon}
\end{bmatrix}\), and $(.)^j$ denotes the $j'th$ power of the matrix.
\end{lemma}

\begin{proof}
It follows from (\ref{estimatere}) that $\boldsymbol{\bar{Q}} = \boldsymbol{\hat{\Theta}^{K_{i+1}}} - {\boldsymbol{\phi_{K_i}}^\top} \boldsymbol{\hat{\Theta}^{K_{i+1}}} \boldsymbol{\phi_{K_i}}$, thus we have:
\begin{equation}
\begin{aligned}
&\sum_{j=0}^\infty \boldsymbol{\phi_{K_i}^\top}^j\boldsymbol{\bar{Q}} \boldsymbol{{\phi_{K_i}}^j}\\
=& \sum_{j=0}^\infty \boldsymbol{\phi_{K_i}^\top}^j[\boldsymbol{\hat{\Theta}^{K_{i+1}}} - {\boldsymbol{\phi_{K_i}}}^\top \boldsymbol{\hat{\Theta}^{K_{i+1}}} \boldsymbol{\phi_{K_i}}] \boldsymbol{\phi_{K_i}^j}\\
=& \sum_{j=0}^\infty \boldsymbol{\phi_{K_i}^\top}^j \boldsymbol{\hat{\Theta}^{K_{i+1}}} \boldsymbol{\phi_{K_i}^j} - \sum_{j=1}^\infty \boldsymbol{\phi_{K_i}^\top}^j \boldsymbol{\hat{\Theta}^{K_{i+1}}} \boldsymbol{\phi_{K_i}^j}\\
=& \boldsymbol{\hat{\Theta}^{K_{i+1}}}.
\end{aligned}
\end{equation}

It can be obtained from (\ref{estimatere}) that
\begin{equation}\label{331}
\boldsymbol{\hat{\Theta}^{K_{i}}} - \boldsymbol{\hat{\Theta}^{K_{i+1}}} = {\boldsymbol{\phi_{K_{i-1}}^\top}} \boldsymbol{\hat{\Theta}^{K_{i}}} \boldsymbol{\phi_{K_{i-1}}} - {\boldsymbol{\phi_{K_i}^\top}} \boldsymbol{\hat{\Theta}^{K_{i+1}}} \boldsymbol{\phi_{K_i}}.
\end{equation}

Note that with \(\boldsymbol{\phi_{K_{i-1}}} = \boldsymbol{\phi_{K_{i}}} + \boldsymbol{\Psi_{K_{i}}}\), we have:
\begin{equation}\label{33}
\begin{aligned}
\boldsymbol{\phi_{K_{i-1}}^\top} \boldsymbol{\hat{\Theta}^{K_i}} \boldsymbol{\phi_{K_{i-1}}} =& (\boldsymbol{\phi_{K_i}} + \boldsymbol{\Psi_{K_i}})^\top \boldsymbol{\hat{\Theta}^{K_i}} (\boldsymbol{\phi_{K_i}} + \boldsymbol{\Psi_{K_i}}) \\
=& \boldsymbol{\phi_{K_i}^\top} \boldsymbol{\hat{\Theta}^{K_i}} \boldsymbol{\phi_{K_i}} + \boldsymbol{\Psi_{K_i}^\top} \boldsymbol{\hat{\Theta}^{K_i}} \boldsymbol{\Psi_{K_i}} \\
&+ \boldsymbol{\phi_{K_i}^\top} \boldsymbol{\hat{\Theta}^{K_i}} \boldsymbol{\Psi_{K_i}} + \boldsymbol{\Psi_{K_i}^\top} \boldsymbol{\hat{\Theta}^{K_i}} \boldsymbol{\phi_{K_i}}.
\end{aligned}
\end{equation}
Substituting $\boldsymbol{K_i}$ based on the update policy in the algorithm $\boldsymbol{K_i} = [\boldsymbol{\hat{\Theta}^{K_i}_{\delta u \delta u}}]^{-1} \boldsymbol{\hat{\Theta}^{K_i}_{\delta u \epsilon}}$, the terms \(\boldsymbol{\phi_{K_i}}^\top \boldsymbol{\hat{\Theta}^{K_i}} \boldsymbol{\Psi_{K_i}} = 0\) and \(\boldsymbol{\Psi_{K_i}}^\top \boldsymbol{\hat{\Theta}^{K_i}} \boldsymbol{\phi_{K_i}} = 0\). Thus (\ref{33}) can be rewritten as:
\begin{equation}\label{replace}
\boldsymbol{\phi_{K_{i-1}}^\top} \boldsymbol{\hat{\Theta}^{K_i}} \boldsymbol{\phi_{K_{i-1}}} = \boldsymbol{\phi_{K_i}^\top} \boldsymbol{\hat{\Theta}^{K_i}} \boldsymbol{\phi_{K_i}} + \boldsymbol{\Psi_{K_i}^\top} \boldsymbol{\hat{\Theta}^{K_i}} \boldsymbol{\Psi_{K_i}}.
\end{equation}
Substituting (\ref{replace}) into (\ref{331}), we can obtain that:
\begin{equation}\label{inter}
\begin{aligned}
\boldsymbol{\hat{\Theta}^{K_i}} - \boldsymbol{\hat{\Theta}^{K_{i+1}}} = &
\boldsymbol{\phi_{K_i}^\top}
\left( \boldsymbol{\hat{\Theta}^{K_i}} - \boldsymbol{\hat{\Theta}^{K_{i+1}}} \right)
\boldsymbol{\phi_{K_i}} \\
&+
\boldsymbol{\Psi_{K_i}^\top}
\boldsymbol{\hat{\Theta}^{K_i}}
\boldsymbol{\Psi_{K_i}}.
\end{aligned}
\end{equation}
For simplicity in notation, we denote \(\boldsymbol{\hat{\Theta}^{K_i}} - \boldsymbol{\hat{\Theta}^{K_{i+1}}}\) as \(\Delta_i\), thus (\ref{inter}) can be rewritten as:
\begin{equation}\label{inter2}
\begin{aligned}
\boldsymbol{\Psi_{K_i}^\top}
\boldsymbol{\hat{\Theta}^{K_i}}
\boldsymbol{\Psi_{K_i}} &=
\Delta_i - \boldsymbol{\phi_{K_i}^\top}
\Delta_i
\boldsymbol{\phi_{K_i}}.
\end{aligned}
\end{equation}
Therefore, (\ref{aa}) can be derived as follows,
\begin{equation}\label{finalconvergence}
\begin{aligned}
&\sum_{j=0}^\infty
\boldsymbol{\phi_{K_i}^\top}^j [\boldsymbol{\Psi_{K_i}}^\top
\boldsymbol{\hat{\Theta}^{K_i}}
\boldsymbol{\Psi_{K_i}}] \boldsymbol{\phi_{K_i}^j}\\
=& \sum_{j=0}^\infty
\boldsymbol{\phi_{K_i}^\top}^j [\Delta_i - \boldsymbol{\phi_{K_i}}^\top
\Delta_i
\boldsymbol{\phi_{K_i}}] \boldsymbol{\phi_{K_i}^j}= \boldsymbol{\hat{\Theta}^{K_i}} - \boldsymbol{\hat{\Theta}^{K_{i+1}}}.
\end{aligned}
\end{equation}

Similarly, we can deduce from (\ref{estimatere}) that:
\begin{equation}\notag
\begin{aligned}
\boldsymbol{\hat{\Theta}^{K_{i+1}} - \Theta_*} =&
\boldsymbol{\phi_{K_i}^\top \hat{\Theta}^{K_{i+1}} \phi_{K_i}} - \boldsymbol{\phi_*^\top \Theta_* \phi_*}\\
=&\boldsymbol{(\phi_*^\top + \Psi_*^\top) \hat{\Theta}^{K_{i+1}} (\phi_* + \Psi_*)} - \boldsymbol{\phi_*^\top \Theta_* \phi_*}\\
=& \boldsymbol{\phi_*^\top (\hat{\Theta}^{K_{i+1}} - \Theta_*) \phi_*} + \boldsymbol{\Psi_*^\top \hat{\Theta}^{K_{i+1}} \phi_*} \\
&+ \boldsymbol{\phi_*^\top \hat{\Theta}^{K_{i+1}} \Psi_*} + \boldsymbol{\Psi_*^\top \hat{\Theta}^{K_{i+1}} \Psi_*},
\end{aligned}
\end{equation}
where \(\boldsymbol{\Psi_*} = \begin{bmatrix}
\boldsymbol{0} & \boldsymbol{0}\\
\boldsymbol{(K_* - K_{i+1}) A_\epsilon} & \boldsymbol{(K_* - K_{i+1}) B_\epsilon}
\end{bmatrix}\). Noting that \(\boldsymbol{-\Psi_*^\top \Theta_* \phi_* = 0}\) and \(\boldsymbol{-\phi_*^\top \Theta_* \Psi_* = 0}\), we have:
\begin{equation}
\begin{aligned}
\boldsymbol{\hat{\Theta}^{K_{i+1}} - \Theta_*} =& \boldsymbol{(\phi_*^\top + \Psi_*^\top) (\hat{\Theta}^{K_{i+1}} - \Theta_*) (\phi_* + \Psi_*)}\\
&+ \boldsymbol{\Psi_*^\top \Theta_* \Psi_*}\\
=&\boldsymbol{\phi_{K_i}^\top (\hat{\Theta}^{K_{i+1}} - \Theta_*) \phi_{K_i}} + \boldsymbol{\Psi_*^\top \Theta_* \Psi_*},
\end{aligned}
\end{equation}
thus proving that \(\boldsymbol{\hat{\Theta}^{K_{i+1}}} - \boldsymbol{\Theta_*} =
\sum_{j=0}^\infty
\left(\boldsymbol{\phi_{K_i}^\top}^j \boldsymbol{\Psi_*^\top \Theta_* \Psi_*} \boldsymbol{\phi_{K_i}^j}\right) \succeq 0\), similar to (\ref{finalconvergence}).
\end{proof}

\begin{theorem}
Under \textbf{Assumptions 1-3}, $\boldsymbol{K_i}$ in \textbf{Algorithm 1} converges elementwise to the optimal controller \(\boldsymbol{K_*}\), defined in \textbf{Proposition 2}, as \(\varepsilon\to0\). 
\end{theorem}

\begin{proof}
It follows from \textbf{Lemma 5} that \(\boldsymbol{\hat{\Theta}^{K_i}} \succeq \boldsymbol{\hat{\Theta}^{K_{i+1}}} \succeq \boldsymbol{\Theta_*}\), thus for any \(\boldsymbol{x \in \mathbb{R}^n}\), it holds that \(\boldsymbol{x^\top \hat{\Theta}^{K_i} x}\) 
is a non-increasing sequence as $i$ increases, bounded below, which hence tends to a limit as $i\to\infty$.
For any \(\boldsymbol{x}, \boldsymbol{y} \in \mathbb{R}^n\), we have:
$\boldsymbol{x^\top \hat{\Theta}^{K_i} y}$$ = \frac{1}{4} \boldsymbol{\left[ (x+y)^\top \hat{\Theta}^{K_i} (x+y) - (x-y)^\top \hat{\Theta}^{K_i} (x-y) \right]}$
Since the right-hand side converges for each pair \(\boldsymbol{(x,y)}\) as \(i \to \infty\), the bilinear form \(\boldsymbol{x^\top \hat{\Theta}^{K_i} y}\) tends to a limit.
Consequently, taking \(\boldsymbol{x}\) and \(\boldsymbol{y}\) as standard basis vectors, we deduce that each individual entry of matrix \(\boldsymbol{\Theta^{K_i}}\) 
tends to a limit point.
Thus, we conclude that the matrix sequence \(\boldsymbol{\Theta^{K_i}}\) converges element-wise to a limiting matrix as $i\to\infty$.

We then show that \(\boldsymbol{\Theta_*}\) in \textbf{Proposition 2} is a unique fixed point of \textbf{Algorithm~1} as \(\varepsilon\to0\) and therefore \(\boldsymbol{\Theta^{K_i}}\to\boldsymbol{\Theta_*}\). Let \(\boldsymbol{K_* = (\Theta_{*\delta u \delta u})^{-1} \Theta_{*\delta u \epsilon}}\), and let \(\boldsymbol{\Theta^{K_*}}\) solve (\ref{estimation}) for \(\boldsymbol{K_{i+1} = K_*}\), therefore, \(\boldsymbol{\hat{\Theta}^{K_{i+1}} = \Theta_*}\). Following this, we show that this fixed point is unique. Assume that \textbf{Algorithm 1} reaches a fixed point such that \(\boldsymbol{\Theta^K}=\boldsymbol{\bar{Q}}+ \boldsymbol{\phi_K}^\top \boldsymbol{\Theta^K}\boldsymbol{\phi_K}\) with \(\boldsymbol{K^F = (\Theta^F_{\delta u \delta u})^{-1} \Theta^F_{\delta u \epsilon}}\). As \(\boldsymbol{\delta u_k = -K^F \epsilon_k}\) is the minimizer of (\ref{minimizationpro}), then \(\boldsymbol{V^{K^F}(\epsilon_k) = \min_{K_k} [c(\epsilon_k, -K_k \epsilon_k) + V^{K^F}(\epsilon_{k+1})]}\). This is the corresponding Hamilton-Jacobi-Bellman equation, which is known to have a unique solution \(\boldsymbol{V^{K^F} = V_*}\), and, therefore, \(\boldsymbol{\Theta^F = \Theta_*}\).
\end{proof}

\begin{lemma}\textbf{(Decreasing \(\boldsymbol{P^{K_i}}\)).}
Let \(\boldsymbol{P^{K_i}}\) and \(\boldsymbol{P^{K_{i+1}}}\) be the solutions of Equation \eqref{bellman1} for the controllers \(\boldsymbol{K_i}\) and \(\boldsymbol{K_{i+1}}\), respectively, at two consecutive iterations of \textbf{Algorithm 1}. It then holds that \(\boldsymbol{P^{K_i} \succeq P^{K_{i+1}}}\).
\end{lemma}

\begin{proof}
Consider the following sequence of inequalities:
\begin{subequations}
    \begin{align}
    \boldsymbol{\epsilon_k^\top P^{K_{i+1}} \epsilon_k} &= \begin{bmatrix}
    \boldsymbol{\epsilon_{k}}^\top & -\boldsymbol{\epsilon_{k}^\top K_{i+1}^\top}
    \end{bmatrix} \boldsymbol{\hat{\Theta}^{K_{i+1}}} \begin{bmatrix}
    \boldsymbol{\epsilon_{k}} \\ -\boldsymbol{K_{i+1} \epsilon_{k}}
    \end{bmatrix}\label{42a}\\
    &\leq \begin{bmatrix}
    \boldsymbol{\epsilon_{k}}^\top & -\boldsymbol{\epsilon_{k}^\top K_{i+1}^\top}
    \end{bmatrix} \boldsymbol{\hat{\Theta}^{K_i}} \begin{bmatrix}
    \boldsymbol{\epsilon_{k}} \\ -\boldsymbol{K_{i+1} \epsilon_{k}}
    \end{bmatrix}\label{42b}\\
    &\leq \begin{bmatrix}
    \boldsymbol{\epsilon_{k}}^\top & -\boldsymbol{\epsilon_{k}^\top K_{i}^\top}
    \end{bmatrix} \boldsymbol{\hat{\Theta}^{K_i}} \begin{bmatrix}
    \boldsymbol{\epsilon_{k}} \\ -\boldsymbol{K_{i} \epsilon_{k}}
    \end{bmatrix}\label{42c}\\
    &= \boldsymbol{\epsilon_k^\top P^{K_{i}} \epsilon_k}\label{42d}.
    \end{align}
\end{subequations}
We derive inequality \eqref{42b} from \eqref{42a} using Lemma~\ref{lem:dec_theta}.
Since $\boldsymbol{K_{i+1}}$ is calculated through the minimization process described in Equation \eqref{minimizationpro}, we can obtain inequality \eqref{42c} from \eqref{42b}.
\end{proof}

\begin{theorem}\label{Thm:conv_eq}
Under \textbf{Assumptions 1-3}, {\bf Algorithm 1} ensures the augmented system (\ref{timingterms}) converges to the equilibrium point \((\bm{\delta u^*_k} = 0, \bm{\epsilon^*_k} = 0, \bm{e^*_k} = 0)\).
\end{theorem}

\begin{proof}
Consider the candidate Lyapunov function as:
$V_k(\bm{\epsilon_k}) = \bm{\epsilon_k}^\top \bm{P^{K_k}} \bm{\epsilon_k}$,
where \(\bm{K_k}\) is the control gain at step \(k\), and note that  \(V_k \geq 0\). Moreover, \(V_k = 0\) if and only if \(\bm{\epsilon_k = 0}\) as $\boldsymbol{P^{K_k}}$ is positive definite.

When Algorithm is \textbf{updating the Controller,}
\[
V_{k+1} - V_k = \bm{\epsilon_{k+1}}^\top \bm{P^{K_{k+1}}} \bm{\epsilon_{k+1}} - \bm{\epsilon_{k}}^\top \bm{P^{K_k}} \bm{\epsilon_{k}} \leq 0,
\]
from \textbf{Lemma 6}. Note also that
\[
\begin{aligned}
    &\bm{\epsilon_{k+1}}^\top \bm{P^{K_{k+1}}} \bm{\epsilon_{k+1}} - \bm{\epsilon_{k}}^\top \bm{P^{K_k}} \bm{\epsilon_{k}}\leq\\
    &\bm{\epsilon_k^\top(A_\epsilon-B_\epsilon K_k)^\top} \bm{P^{K_{k}}} \bm{(A_\epsilon-B_\epsilon K_k) \epsilon_k} -\bm{\epsilon_{k}}^\top \bm{P^{K_k}} \bm{\epsilon_{k}} \leq 0,
\end{aligned}
\]
where equality with zero holds if and only if \(\bm{\epsilon_k =0}\).

When the controller \textbf{remains unchanged},
\[
\begin{aligned}
    V_{k+1} - V_k &= \bm{\epsilon_{k+1}}^\top \bm{P^{K_k}} \bm{\epsilon_{k+1}} - \bm{\epsilon_{k}}^\top \bm{P^{K_k}} \bm{\epsilon_{k}} \\
    &= -\bm{\epsilon_k}^\top \bm{Q_{\textbf{eff}}} \bm{\epsilon_k} \leq 0.
\end{aligned}
\]
Here, \(V_k\) is non-increasing and bounded below by zero, and convergence to the equilibrium point specified in {\bf Theorem \ref{Thm:conv_eq}} follows from the discrete-time Lasalle's Invariance Principle. In particular, the system trajectory converges to the largest invariant set contained within \(\{\bm{\epsilon_k} \mid V_{k+1} = V_k\}\). This set is \(\{\bm{\epsilon_k} \mid \bm{\epsilon_k = 0}\}\) indicating convergence to 
the desired equilibrium point where we also have \(\bm{\delta u_k = 0},\bm{e_k} = 0\).
\end{proof}
\begin{remark}
It should be noted from the proof of {\bf Theorem~\ref{Thm:conv_eq}} that convergence to the desired equilibrium point is also ensured if  {\bf Algorithm 1} is not terminated, i.e. step 8 is not included 
and the controller $K_i$ is continuously updated.
\end{remark}

\section{Numerical Experiments}
We have evaluated the performance of the proposed energy-sharing temperature regulation scheme through simulations conducted on DHSs. These simulations accounted for the effects of disturbances and model variations.

\subsection{Coefficient variation}
\xinyiafterreview{We validated the DHS controller in a North China industrial park with three producers and eight loads \cite{yi2023energy}. The study examined coefficient variations in the mass flow rate matrix \(\boldsymbol{L_q}\) from its nominal value \(\boldsymbol{\hat{L}_q}\) and employed probing noise that decays over time. The probing noise is generated by combining sinusoidal functions with varying frequencies and phases. Although the model assumes constant mass flow, practical adjustments often respond rapidly to demand changes, leading to an uncertain system matrix. Table~\ref{compare1} compares the proposed online and optimal controllers, highlighting differences in control gain \(\boldsymbol{K}\) and convergence under model deviations. The Difference row reports the normalized error \(\frac{\|K_k-K_{\text{optimal}}\|_F}{\|\boldsymbol{K_{\text{optimal}}}\|_F}\), where \(\|\cdot\|_F\) denotes the Frobenius norm. The online controller converges to the optimal control gain \(\boldsymbol{K_*}\) within a few iterations, even under model variations of up to \(\pm 50\%\).} 
\vspace{-0.4cm}
\begin{table}[ht]
\centering
\caption{\xinyiafterreview{Algorithm performance with different model variations}}
\label{compare1}
\begin{tabular}{p{2.1cm}|p{0.7cm}|p{0.7cm}|p{0.7cm}|p{0.6cm}|p{0.6cm}|p{0.6cm}}
\hline
\textbf{Variation} &\textbf{-50\%}& \textbf{-20\%} & \textbf{-10\%} & \textbf{10\%} & \textbf{20\%} & \textbf{50\%}   \\
\hline
\textbf{Difference(e-06)} & 4.9878&4.9846 & 4.9842& 4.9855& 4.9859&4.9880\\
\textbf{Iteration} & 4& 4 & 3 & 3 & 3&3 \\
\hline
\end{tabular}
\end{table}
\vspace{-1.4cm}
  \begin{figure}[htbp]
	\vspace{0.8cm}
	\hspace{-0.1cm}
	\includegraphics[width=3.5in]{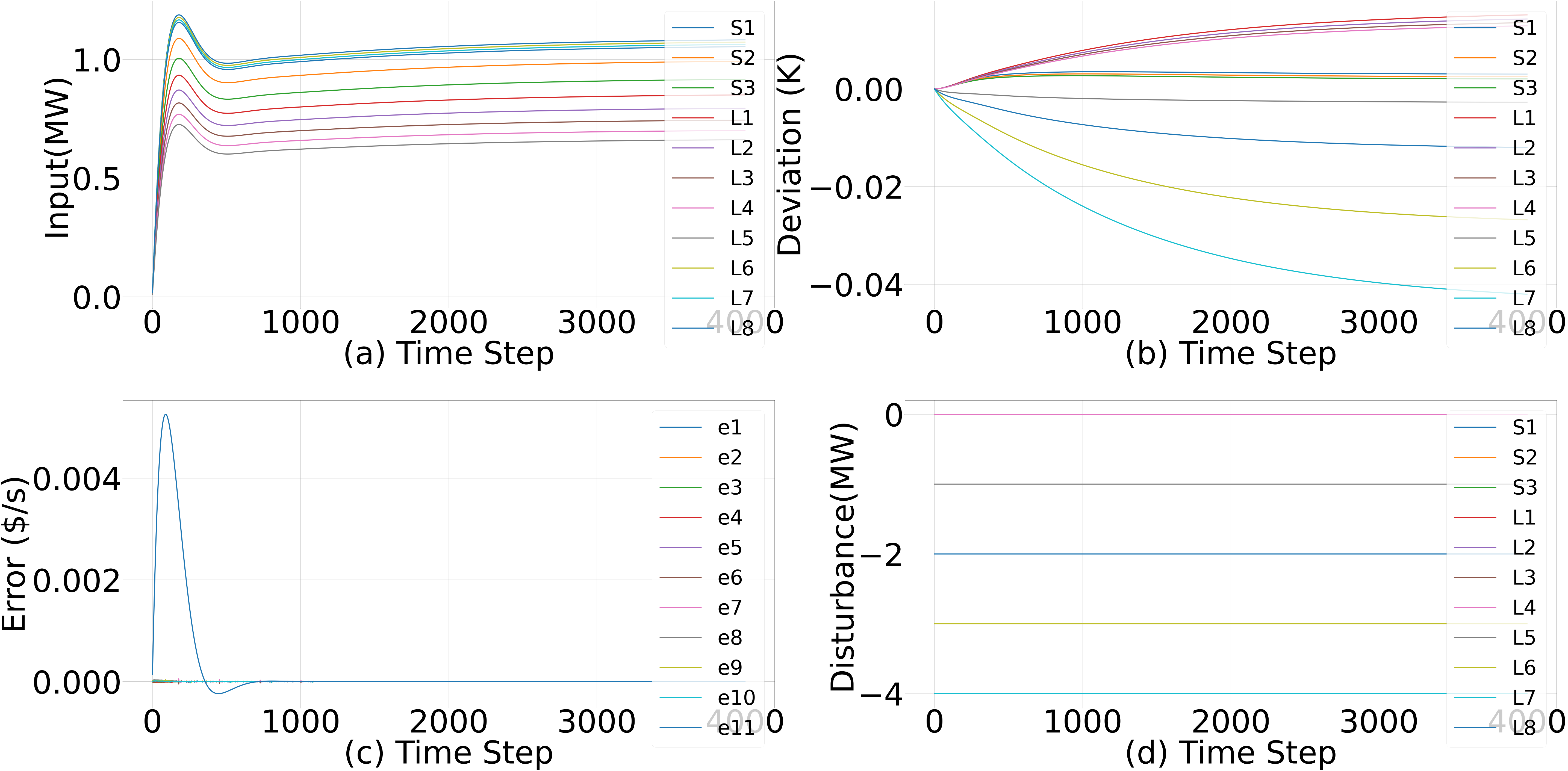}
	\vspace{-0.55cm}
	\caption{(a) Controlled input, (b) Temperature deviation, (c) Error, and (d) Load disturbances of $0MW$, $0MW$, $0MW$, $0MW$, $0MW$, $0MW$, $0MW$, $-1MW$, $-3MW$, $-4MW$, $-2MW$ in the case of $-50\%$ variation.}
	\label{industrialplus}
\end{figure}
\vspace{-0.3cm}

For illustration, we consider a \(-50\%\) variation from the nominal model parameters. Fig.~\ref{industrialplus} shows the controlled inputs and temperature deviations converging to the optimal economic dispatch defined in \eqref{opt1} and \eqref{economic2}. The error corresponds to the integrated optimality-condition error in (\ref{timingtermsb}), which converges to zero, confirming satisfaction of the condition. The algorithm identified the optimal controller in just four iterations. With \(N=275\) from the rank condition in \textbf{Lemmas 2,3}, the data-based scheme provided controller updates at prescribed instances \(i\) during \(k=1\)–\(1100\), after which the system operated with the fixed optimal controller for \(k=1101\)–\(4000\).

\subsection{\xinyiafterreview{The effect of impulse-like and piece-wise disturbances}}
\xinyiafterreview{We evaluated the \xinyiafterrevieww{system} response to sudden spikes and piecewise disturbances, as shown in Fig.~\ref{industrial2}. Disturbances ranged from \((-0.08, 0.09)\,\text{MW}\) at \(k=10\) to \((-0.8, 0.9)\,\text{MW}\) during \(k=7000\!-\!7100\) for each HX. Despite the large error at \(k=10\), which complicates \xinyiafterrevieww{the} Q-function matrix computation, \textbf{Algorithm~1} identified a near-optimal controller within 25 iterations. The controller was updated during \(k=1\!-\!6875\) and fixed for \(k=6876\!-\!10000\). The system converged to the optimal equilibrium with tracking error \(\boldsymbol{e}=0\). Incorporating error dynamics into the design improves robustness, enabling rapid recovery and stable performance under disturbances.}\footnote{\xinyiafterreview{Clock-wise and step-wise disturbances degrade Q-matrix estimation, highlighting the need for robust learning strategies left for future work.}}

\subsection{Comparison with nominal stabilizing controller}
As the initial stabilizing controller drives the augmented system to the equilibrium state \((\boldsymbol{\delta u_k^* = 0}, \boldsymbol{\epsilon_k^* = 0})\), we compared the stage costs of the on-policy and constant nominal controllers (Fig.~\ref{comparecost1}). Initially, both showed similar costs, with the on-policy controller relying on small probing noise before its first update at step 275. After this update, it achieved noticeably lower costs than the nominal controller, particularly between steps 275 and 1100, despite occasional spikes from probing noise.\footnote{\xinyiafterreview{These cost spikes are due to probing noise, which consists of random combinations of trigonometric functions.}} Beyond step 1100, the system operated under a near-optimal trajectory-driven on-policy controller, consistently outperforming the nominal controller. \xinyiafterreview{Future work will examine the effect of probing noise on performance and strategies to mitigate it while preserving PE.}

\vspace{-0.8cm}
  \begin{figure}[htbp]
	\vspace{0.6cm}
	\hspace{-0.1cm}
	\includegraphics[width=3.5in]{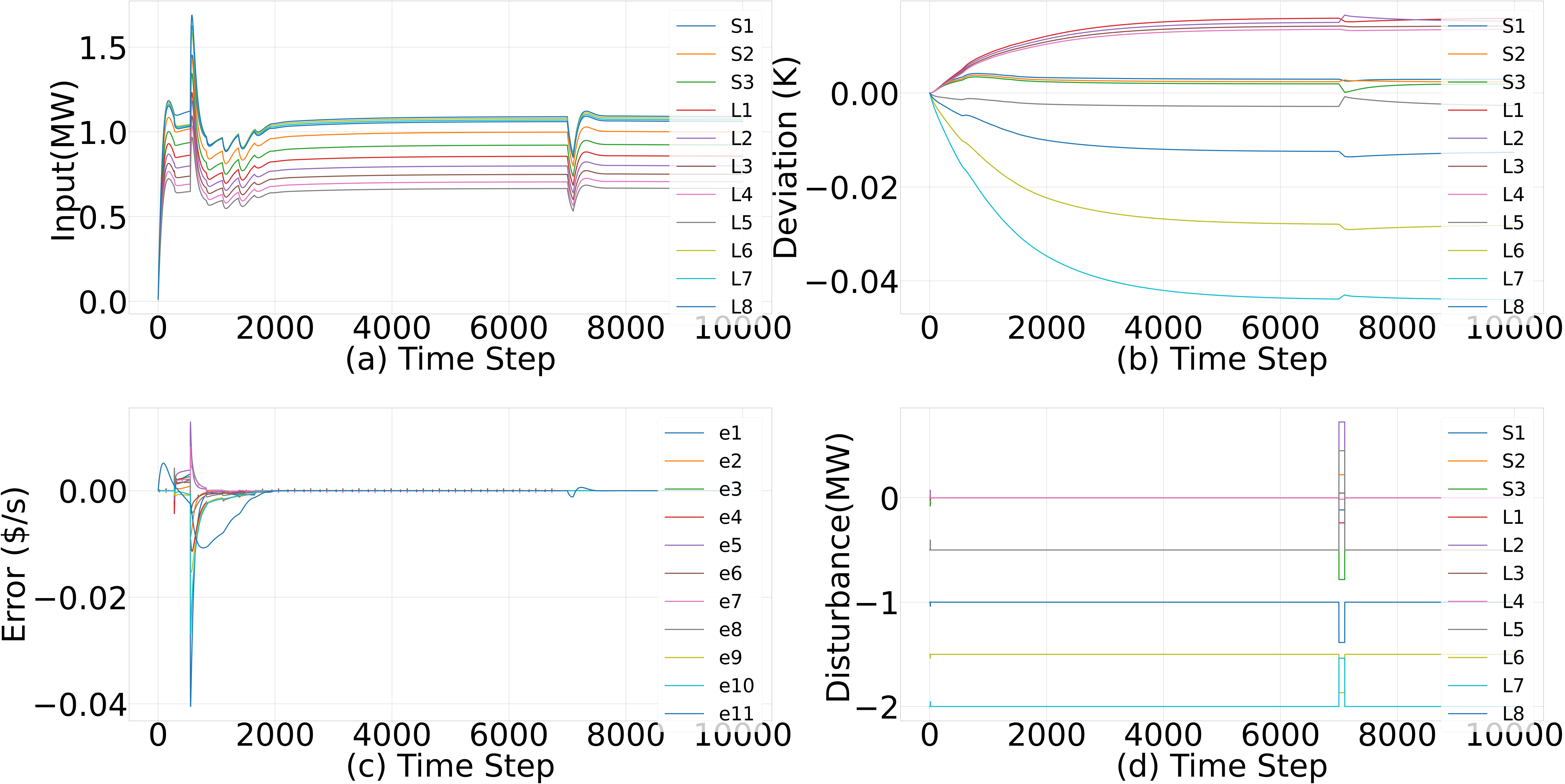}
	\vspace{-0.6cm}
	\caption{\xinyiafterreview{(a) Control input, (b) Temperature deviation, (c) Optimality condition related
error, and (d) Disturbances of DHS with sudden spikes.}}
	\label{industrial2}
\end{figure}
\vspace{-1.4cm}

  \begin{figure}[htbp]
	\vspace{0.6cm}
	\hspace{-0.1cm}
	\includegraphics[width=3.5in]{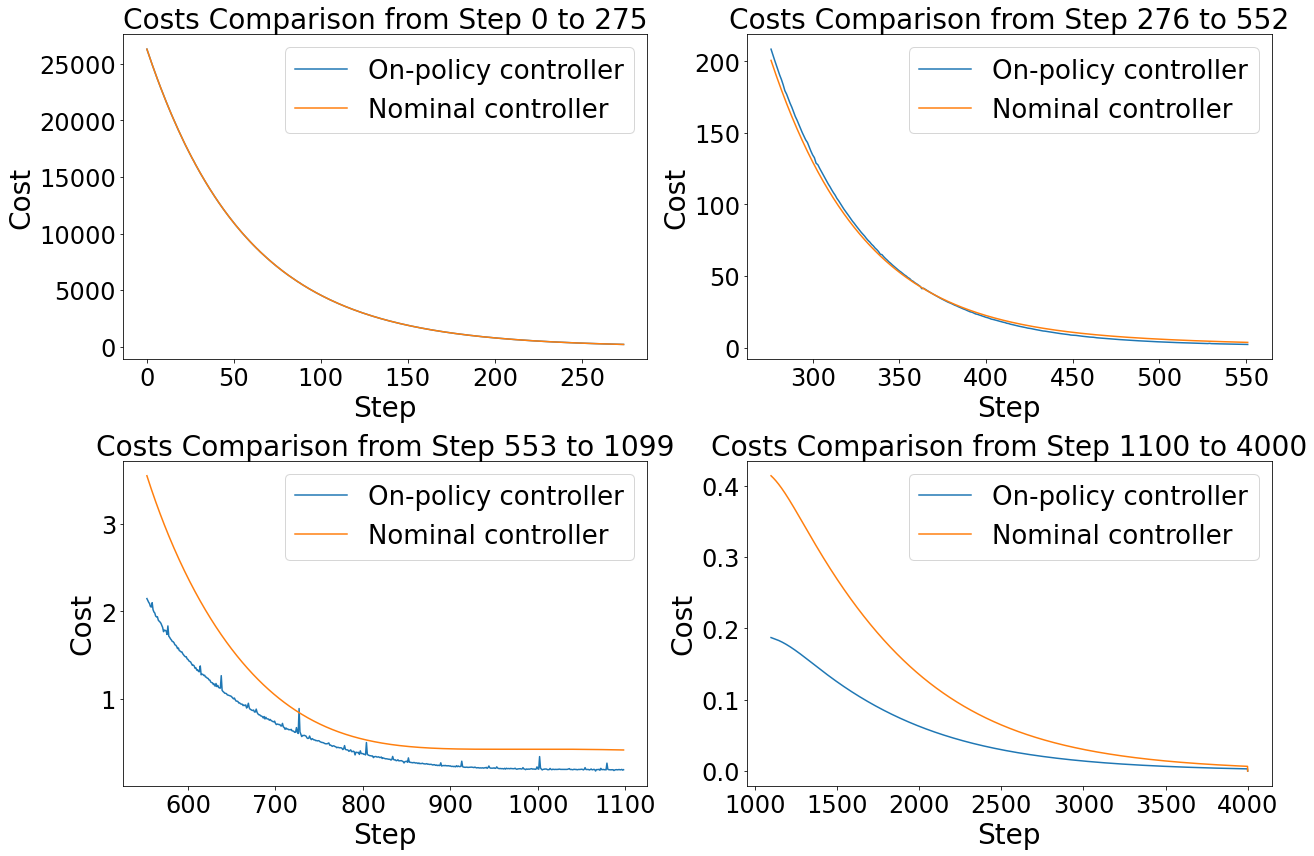}
	\vspace{-0.8cm}
	\caption{Stage costs comparison of nominal controller and on-policy controller}
	\label{comparecost1}
\end{figure}
\vspace{-0.2cm}

\subsection{Comparison with indirect controller}
The indirect (ID) method solves the LQR problem from a data-identified model, whereas the direct method avoids model identification and is more resilient to bias errors such as model order mismatch or nonlinearities \cite{lopez2023efficient}. When data are collected from a nonlinear system $\boldsymbol{\epsilon_{k+1}} = \boldsymbol{A_\epsilon \epsilon_k + B_\epsilon \delta u_k + w \cdot \epsilon_k^2}$, where $\boldsymbol{\epsilon_k^2}$ is the element-wise square of $\boldsymbol{\epsilon_k}$, the augmented system (\ref{timingterms}) remains nearly linear with mild nonlinear perturbations. As shown in Table~\ref{compare3}, even small nonlinearities (e.g., $w=10^{-7}$) significantly degraded ID performance due to bias, while the proposed method remained robust. With stronger nonlinearities (e.g., $w=10^{-4}$, Table~\ref{compare2}), the proposed method required more iterations to converge but still achieved high accuracy, whereas ID exhibited a clear loss of precision.

\vspace{-0.5cm}
\begin{table}[ht]
\centering
\caption{Comparison with different model variations ($w=e-07$)}\label{compare3}
\begin{tabular}{p{2.5cm}|p{0.7cm}|p{0.7cm}|p{0.7cm}|p{0.6cm}}
\hline
\textbf{Variation} & \textbf{-20\%} & \textbf{-10\%} & \textbf{10\%} & \textbf{20\%} \\
\hline
\textbf{RL Difference(e-06)} & 4.984&4.984 & 4.986& 4.986\\
\textbf{RL Iteration} & 4& 3 & 3 & 3 \\
\textbf{ID Difference(e-03)} & 1.200&1.241 &1.394& 1.490\\
\hline
\end{tabular}
\end{table}
\vspace{-0.7cm}
\begin{table}[ht]
\centering
\caption{Comparison with different model variations ($w=e-04$)}\label{compare2}
\begin{tabular}{p{2.5cm}|p{0.7cm}|p{0.7cm}|p{0.7cm}|p{0.6cm}}
\hline
\textbf{Variation} & \textbf{-20\%} & \textbf{-10\%} & \textbf{10\%} & \textbf{20\%} \\
\hline
\textbf{RL Difference(e-06)} & 5.025&4.997 & 4.991& 4.987\\
\textbf{RL Iteration} & 24& 21 & 16 & 15 \\
\textbf{ID Difference} & 0.020&0.017 &0.014& 0.013\\
\hline
\end{tabular}
\end{table}
\vspace{-0.3cm}
\section{Conclusion}
This paper presents an on-policy RL temperature regulator for DHSs, 
capable of achieving steady-state optimality without prior knowledge of disturbances or system models, while also having guaranteed convergence to an optimal controller. The effectiveness of the proposed controller is demonstrated through case studies across a range of scenarios.
\nocite{*}

\bibliographystyle{IEEEtran}
\bibliography{CDCarxivfull} 

\begin{thebibliography}{10}
\providecommand{\url}[1]{#1}
\csname url@samestyle\endcsname
\providecommand{\newblock}{\relax}
\providecommand{\bibinfo}[2]{#2}
\providecommand{\BIBentrySTDinterwordspacing}{\spaceskip=0pt\relax}
\providecommand{\BIBentryALTinterwordstretchfactor}{4}
\providecommand{\BIBentryALTinterwordspacing}{\spaceskip=\fontdimen2\font plus
\BIBentryALTinterwordstretchfactor\fontdimen3\font minus \fontdimen4\font\relax}
\providecommand{\BIBforeignlanguage}[2]{{%
\expandafter\ifx\csname l@#1\endcsname\relax
\typeout{** WARNING: IEEEtran.bst: No hyphenation pattern has been}%
\typeout{** loaded for the language `#1'. Using the pattern for}%
\typeout{** the default language instead.}%
\else
\language=\csname l@#1\endcsname
\fi
#2}}
\providecommand{\BIBdecl}{\relax}
\BIBdecl

\bibitem{liu2024diversifying}
S.~Liu, Y.~Guo, F.~Wagner, H.~Liu, R.~Y. Cui, and D.~L. Mauzerall, ``Diversifying heat sources in china’s urban district heating systems will reduce risk of carbon lock-in,'' \emph{Nature Energy}, vol.~9, no.~8, pp. 1021--1031, 2024.

\bibitem{accconfer}
X.~Yi and I.~Lestas, ``Optimal energy-sharing and temperature regulation in district heating systems,'' \emph{NecSys, IFAC-PapersOnLine}, vol.~59, no.~4, pp. 67--72, 2025.

\bibitem{6315769}
F.~L. Lewis, D.~Vrabie, and K.~G. Vamvoudakis, ``Reinforcement learning and feedback control: Using natural decision methods to design optimal adaptive controllers,'' \emph{IEEE Control Systems Magazine}, vol.~32, no.~6, pp. 76--105, 2012.

\bibitem{735224}
S.~Bradtke, B.~Ydstie, and A.~Barto, ``Adaptive linear quadratic control using policy iteration,'' in \emph{Proceedings of 1994 American Control Conference - ACC '94}, vol.~3, 1994, pp. 3475--3479 vol.3.

\bibitem{lopez2023efficient}
V.~G. Lopez, M.~Alsalti, and M.~A. M{\"u}ller, ``Efficient off-policy q-learning for data-based discrete-time lqr problems,'' \emph{IEEE Transactions on Automatic Control}, vol.~68, no.~5, pp. 2922--2933, 2023.

\bibitem{hao2024quadratic}
L.~Hao, C.~Wang, and Y.~Shi, ``Quadratic tracking control of linear stochastic systems with unknown dynamics using average off-policy q-learning method,'' \emph{Mathematics}, vol.~12, no.~10, p. 1533, 2024.

\bibitem{cholewa2022easy}
T.~Cholewa, A.~Siuta-Olcha, A.~Smolarz, P.~Muryjas, P.~Wolszczak, {\L}.~Guz, M.~Bocian, and C.~A. Balaras, ``An easy and widely applicable forecast control for heating systems in existing and new buildings: First field experiences,'' \emph{Journal of Cleaner Production}, vol. 352, p. 131605, 2022.

\bibitem{machado2022decentralized}
J.~E. Machado, J.~Ferguson, M.~Cucuzzella, and J.~M. Scherpen, ``Decentralized temperature and storage volume control in multiproducer district heating,'' \emph{IEEE Control Systems Letters}, vol.~7, pp. 413--418, 2022.

\bibitem{hangos1999thermodynamic}
K.~M. Hangos, A.~A. Alonso, J.~D. Perkins, and B.~E. Ydstie, ``Thermodynamic approach to the structural stability of process plants,'' \emph{AIChE journal}, vol.~45, no.~4, pp. 802--816, 1999.

\bibitem{asuk2021feedback}
A.~Asuk and P.~Trodden, ``Feedback optimizing linear quadratic control,'' in \emph{2021 American Control Conference (ACC)}.\hskip 1em plus 0.5em minus 0.4em\relax IEEE, 2021, pp. 3800--3805.

\bibitem{1098829}
D.~Kleinman, ``On an iterative technique for riccati equation computations,'' \emph{IEEE Transactions on Automatic Control}, vol.~13, no.~1, pp. 114--115, 1968.

\bibitem{yi2023energy}
X.~Yi, Y.~Guo, H.~Sun, X.~Qin, and Q.~Wu, ``Energy-grade double pricing for combined heat and power systems,'' \emph{IEEE Transactions on Power Systems}, 2023.

\bibitem{netzerocarbon}
{Energy Saving Trust}, ``{What is net zero and how are the UK and other countries doing? },'' \emph{{BBC news}}, 2024.

\bibitem{10383604}
L.~Sforni, G.~Carnevale, I.~Notarnicola, and G.~Notarstefano, ``On-policy data-driven linear quadratic regulator via combined policy iteration and recursive least squares,'' in \emph{2023 62nd IEEE Conference on Decision and Control (CDC)}, 2023, pp. 5047--5052.

\bibitem{pappas1980numerical}
S.~Boyd, ``Lecture 13. linear quadratic lyapunov theory,'' \url{https://stanford.edu/class/ee363/lectures/lq-lyap.pdf}, Stanford University, 2009, accessed: 2023-03-28.

\bibitem{van2020data}
H.~J. Van~Waarde, J.~Eising, H.~L. Trentelman, and M.~K. Camlibel, ``Data informativity: A new perspective on data-driven analysis and control,'' \emph{IEEE Transactions on Automatic Control}, vol.~65, no.~11, pp. 4753--4768, 2020.

\bibitem{kiumarsi2017h}
B.~Kiumarsi, F.~L. Lewis, and Z.-P. Jiang, ``H infinity control of linear discrete-time systems: Off-policy reinforcement learning,'' \emph{Automatica}, vol.~78, pp. 144--152, 2017.

\end{thebibliography}
\end{document}